\documentclass[a4paper,UKenglish,cleveref, autoref, thm-restate]{lipics-v2021}

\pdfoutput=1 
\hideLIPIcs  

\include{include}

\DeclareMathOperator{\poly}{\textnormal{\textsf{poly}}}
\DeclareMathOperator{\polylog}{\textnormal{\textsf{polylog}}}

\renewcommand{\epsilon}{\varepsilon} 

\newcommand{\UU}{\mathcal{U}}

\newcommand{\calP}{\mathcal{P}}
\newcommand{\calQ}{\mathcal{Q}}
\newcommand{\calI}{\mathcal{I}}
\newcommand{\calT}{\mathcal{T}}
\newcommand{\calF}{\mathcal{F}}

\newcommand{\YES}{\textsf{YES}}
\newcommand{\LB}{\textsf{LB}}

\renewcommand{\epsilon}{\varepsilon}
\newcommand{\eps}{\varepsilon}

\newcommand{\size}[1]{\ensuremath{\left|#1\right|}}
\newcommand{\set}[1]{\left\{ #1 \right\}}

\newcommand{\parentheses}[1]{\left(#1\right)}
\newcommand{\expectation}[2]{\mathbb{E}_{#1}\left[ #2 \right]}
\newcommand{\variance}[2]{\text{Var}_{#1}\left[ #2 \right]}
\renewcommand{\Pr}[1]{{\mathrm{Pr}}\left[ #1 \right]}
\newcommand{\Prr}[2]{\mathrm{Pr}_{#1}\left[ #2 \right]}

\newcommand{\ceil}[1]{\left\lceil {#1} \right\rceil}
\newcommand{\indicator}[1]{\left[ {#1} \right]}

\newif\ifnotes
\notestrue
\ifnotes
\newcommand{\ioana}[1]{{\ifnotes \scriptsize \textcolor{red}{Ioana: {#1}} \fi}}
\newcommand{\jakob}[1]{{\ifnotes \scriptsize \textcolor{purple}{Jakob: {#1}} \fi}}
\newcommand{\rasmus}[1]{{\ifnotes \scriptsize  \textcolor{blue} {Rasmus: {#1}} \fi}}
\else
\newcommand{\ioana}[1]{}
\newcommand{\jakob}[1]{}
\newcommand{\rasmus}[1]{}
\fi

\title{Daisy Bloom Filters}

\author{Ioana O. Bercea}{KTH Royal Institute of Technology, Stockholm, Sweden}{bercea@kth.se}{https://orcid.org/0000-0001-8430-2441}{}
\author{Jakob Bæk Tejs Houen}{BARC, University of Copenhagen, Denmark}{jakob@tejs.dk}{https://orcid.org/0000-0002-8033-2130}{}
\author{Rasmus Pagh}{BARC, University of Copenhagen, Denmark}{pagh@di.ku.dk}{https://orcid.org/0000-0002-1516-9306}{}

\authorrunning{I.\,O. Bercea, J.\,B.\,T. Houen, R. Pagh} 

\Copyright{Ioana O. Bercea, Jakob Bæk Tejs Houen, and Rasmus Pagh } 



\category{} 

\funding{Supported by grant 16582,  Basic Algorithms Research Copenhagen (BARC), from the VILLUM Foundation.}

\EventEditors{}
\EventNoEds{0}
\EventLongTitle{}
\EventShortTitle{}
\EventAcronym{}
\EventYear{}
\EventDate{}
\EventLocation{}
\EventLogo{}
\SeriesVolume{}
\ArticleNo{}

\begin{document}
	\nolinenumbers
	
	\pagenumbering{roman}
	\maketitle
	
	\begin{abstract}
		A filter  is a widely used data structure for storing an approximation of a given set $S$ of elements from some universe $\UU$ (a countable set).
		It represents a superset $S'\supseteq S$ that is ``close to~$S$'' in the sense that for $x\not\in S$, the probability that $x\in S'$ is bounded by some $\eps > 0$. The advantage of using a Bloom filter, when some false positives are acceptable, is that the space usage becomes smaller than what is required to store $S$ exactly. 
		
		Though filters are well-understood from a worst-case perspective, it is clear that state-of-the-art constructions may not be close to optimal for particular distributions of data and queries. Suppose, for instance, that some elements are in $S$ with probability close to~1. Then it would make sense to always include them in $S'$, saving space by not having to represent these elements in the filter. Questions like this have been raised in the context of Weighted Bloom filters  (Bruck, Gao and Jiang, ISIT 2006) and Bloom filter implementations that make use of access to learned components (Vaidya, Knorr, Mitzenmacher, and Krask, ICLR 2021).
		
		In this paper, we present a lower bound for the expected space that such a filter requires. We also show that the lower bound is asymptotically tight  by exhibiting a filter construction that executes queries and insertions in worst-case constant time,  and has a false positive rate at most $\eps$ with high probability over input sets drawn from a product distribution. We also present a Bloom filter alternative, which we call the \emph{Daisy Bloom filter}, that executes operations faster and uses significantly less space than the standard Bloom filter.  
	\end{abstract}
	
	\newpage
	\pagenumbering{arabic}

	\section{Introduction}
	This paper shows asymptotically matching upper and lower bounds for the space of an optimal (Bloom) filter when the input and queries come from specific distributions. For a set $S$ of keys (the input set), a filter on $S$ with parameter $\epsilon \in (0,1)$ is a data structure that answers membership queries of the form ``is x in S?'' with a one-sided error: if $x \in S$, then the filter always answers \textsf{YES}, otherwise it makes a mistake (i.e., a false positive) with probability at most $\epsilon$.  The Bloom filter~\cite{bloom1970space} is the most widely known such filter, although more efficient constructions are known~\cite{bender2011don, fan2014cuckoo, bender2018bloom,BerceaE20,BerceaE21, dietzfelbinger2008succinct, liu2020succinct, DBLP:conf/soda/PaghPR05, pagh2013approximate, porat2009optimal}. Filters are also intimately related to dictionaries (or hash tables), the latter of which always answer membership queries exactly.
	
	When errors can be tolerated, (Bloom) filters are much better than dictionaries at encoding the input set: they require $ \Theta( n\log(1/\epsilon))$ bits to represent a set of size $n$, versus the $\geq n\log(u/n)$ bits that a dictionary would require (here, $u$ is the size of the universe). As such, filters are often used in conjunction with dictionaries to speed up negative queries. In particular, filters are often stored in a fast but small memory and are used to ``filter out'' a majority of negative queries to a dictionary (which might reside in big but slow memory).  Because of this, they have proved to be extremely popular in practice and research on them continues to this day, both in the direction of practical implementations~\cite{DBLP:conf/sigmod/PandeyCDBFJ21, DBLP:journals/pvldb/EvenEM22,  DBLP:journals/pacmmod/DayanBRP23} and on the theoretical front~\cite{ liu2020succinct,BerceaE20,BenderFKKL22}. For instance, recent advances in filter design have included making them dynamic, resizeable and lowering the overall space that they require.

	\subparagraph*{The filter encoding.}  In this paper, we ask ourselves what should optimal filters look like when they encode sets that come from a specific distribution. While this question has been resolved for exact encodings (i.e., entropy), no similar concepts are known for filter encodings. Indeed, considering input distributions raises several technical questions. For instance, it is not even clear how to define the concept of approximate membership with respect to a set drawn from a distribution. Should we assume that the input set is given to us in full before we build our filter and allocate memory? 	Moreover, we would like to obtain designs that are never worse than filters with no knowledge of the input distribution, both in space allocated and time required to perform every operation. Should we then require that the false positive guarantee hold for every possible input set or just on average over  the input distribution?

	We also study optimality when additionally, we have access to a distribution over queries. 
	This is especially important for applications in which the performance of the filter is measured over a sequence of queries, rather than for each query separately~\cite{fan2000summary, broder2003network}. 
	At the extreme end of this one can consider adversarial settings, in which an adversary forces the filter to incur many false positives (which can cause a delay in the system by forcing the filter to repeatedly access the slow dictionary). 
	In these settings, defining what it means for the filter to behave efficiently can be a challenge and several definitions have been considered~\cite{naor2015bloom, mitzenmacher2018adaptive, bender2018bloom, naor2022bet}. 
	For us, the challenge is to use the query distribution to obtain gains, while making sure that the filter does not on average exhibit more false positives than usual.
	This is natural when each false positive has the same cost, independent of the query element.

	To this end,  we consider a natural generative model of input sets and queries. Specifically, we let $\calP$ and $\calQ$ denote two distributions over the universe $\UU$ of keys and let $p_x$ (and $q_x$, respectively) denote the probability that a specific key $x\in\UU$ is sampled from $\calP$ (and $\calQ$, respectively). 
	The input set $S$ is generated by $n$ independent draws (with replacement) from $\mathcal{P}$ and we let $\calP_n$ denote this product distribution.\footnote{We do not consider multiplicities although our design can be made to handle them by using techniques from counting filters~\cite{DBLP:conf/soda/PaghPR05, bonomi2006improved, pandey2017general, BerceaE21}.} 
	
	We then define approximate membership for a fixed set $S$ to mean that the average false positive probability over  $\calQ$ is at most $\eps$. Specifically, let $\calF$ denote the filter and  let $\calF(S,x) \in \set{\textsf{YES},\textsf{NO}}$ denote the answer that $\calF$ returns when queried on an element $x\in\UU$, after having been given $S\subseteq \UU$ as input. Then we propose the following definition:
	
	\begin{definition}
		For any $\eps$ with $0<\eps<1$, we say that $\calF$ is a \emph{$(\calQ,\eps)$-filter for $S$} if it satisfies the following conditions:
		\begin{enumerate}
			\item No false negatives: For all $x\in S$, we have that $\Pr{\calF(S,x) = \textsf{YES}} = 1$.
			\item Bounded false positive rate: 
			\begin{center}
				$\sum_{x\in \UU\setminus S} q_x \cdot \Pr{\calF(S,x) = \textsf{YES}} \leq \eps$
			\end{center}
			
		\end{enumerate}
	\end{definition}
	
	We note a detail in the above definition that has important technical consequences and that is, the false positive rate is not computed with respect to the input distribution (i.e., the probability of a false positive only depends on the internal randomness of the filter and not the random process of drawing the input set).  
	As a consequence, we can argue about filter designs that work over  all input sets except some that occur very rarely under $\calP_n$. This is stronger than saying that $\calF$ works only on average over $\calP_n$. Moreover,  we also want designs that do not require knowing the specific realization of the input set in advance. Our dependency on $\calP_n$ shows up in the space requirements of the filter.
	
	\subparagraph*{Access to $\calP$ and $\calQ$. }  For simplicity, we consider filter designs that have oracle access to $\calP$ and $\calQ$: upon seeing a key $x$, we also get $p_x$ and $q_x$. We assume that this is done in constant time and do not account for the size of the oracle when we bound the size of the filter.  Critics of this model have argued that assuming oracle access to a distribution over the universe is too strong of an assumption. Indeed, this is a valid concern, since we are talking about a data structure that is meant to save space over a dictionary. We try to alleviate this concern in several ways. On one hand, our construction can tolerate mistakes. In particular, our designs are robust even if we have a constant factor approximation for $p_x$ and $q_x$, in the sense in which the space increases only by $O(n)$ bits and the time to perform each operation by an added constant. The assumption of access to such approximate oracles is standard~\cite{canonne2014testing, eden2020learning} and can be based on samples of historical information, on frequency estimators such as Count-Min~\cite{cormode2005improved} or Count-Sketch~\cite{charikar2002finding}, or on machine learning models (see for instance, the neural-net based frequency predictor of Hsu et al.~\cite{hsu2019learning}). This view is indeed part of an emerging body of work on algorithms with predictions, to which the data structure perspective is just beginning to contribute~\cite{NEURIPS2018_0f49c89d, cao2023learningaugmented, ferragina2020pgm,FerraginaV19, pmlr-v119-ferragina20a, DBLP:conf/esa/FerraginaL0V23, vaidya2021partitioned, dai2020adaptive, mccauley2023online}.
	
	On the other hand, empirical studies have shown that significant gains are possible even when using off-the-shelf, ``simplistic'' learned components such as random forest classifiers. In particular,   the Partitioned Learned Bloom Filter~\cite{vaidya2021partitioned} and the Adaptive Learned Bloom Filter~\cite{dai2020adaptive} consider settings in which the size of the learned component is comparable to the size of the filter itself (rather than proportional to the size of the universe), and compare the traditional Bloom filter design~\cite{bloom1970space} with a learned design whose space includes the random forest classifier. In one experiment with a universe of $\approx 138,000$ keys and a classifier of $136$Kb,~\cite{dai2020adaptive}  show that, within the range $150$-$300$Kb, there is a $98\%$ decrease in false positive rate compared to the original Bloom filter. This continues to hold for larger universe ($\approx$ $450,000$ keys) with total allocated space between $200$Kb and $1000$Kb. A discussion of how our current (theoretical) design compares to the ones in~\cite{dai2020adaptive} and~\cite{vaidya2021partitioned} can be found in Sec.~\ref{sec:related}.

	Finally, strictly speaking, our designs do not necessarily rely on knowing $p_x$ and $q_x$ for every element inserted or queried. As we will see in the next section, our designs depend rather on knowing which subset of the universe a key $x$ belongs to. This corresponds to a partitioning of the universe that mainly depends on the ratio $q_x/p_x$, rather than the individual values of $p_x$ and $q_x$ (with the exception of values of $p_x$ and $q_x$ that are very small, e.g., smaller than $1/n$). This can conceivably lead to even smaller oracles that just output the partition to which an element belongs. We also do not need to query the entire universe in order to set the internal parameters of the filter, in contrast to~\cite{dai2020adaptive,vaidya2021partitioned}.
	
	\subparagraph*{Weighted Bloom filters.} 
	The design that we propose starts by gathering information about the input and query distributions, using $\polylog(n)$ samples.\footnote{Elements that are inserted in the set during that time can be stored in a small dictionary that only requires $\polylog(n)$ bits, see Sec.~\ref{sec:estimate}.} This information is used  to estimate the internal parameters of the filter which are then used to allocate space for the filter and implement the query and insert operations. Thus, the most important aspect of our design is in setting the aforementioned internal parameters. 
	
	As a baseline for comparison, we can consider the classic Bloom filter design which allocates an array of $\approx 1.44 \cdot n  \log(1/\eps)$ bits and hashes every key to $\log(1/\eps)$ locations in the array. Upon insertion, the corresponding bits are set to 1 and a query returns a YES if and only if all locations are set to 1.  The more locations we hash into, the lower the probability that we make a mistake. Thus, a natural approach for our problem would be to vary the number of hashed locations of $x$ based on $p_x$ and $q_x$. Indeed, this is the question investigated by Bruck, Gao and Jiang~\cite{BruckGJ06weighted} in their Weighted Bloom filter design. More precisely, let $k_x$ denote the number of locations that key $x$ is hashed to. Then~\cite{BruckGJ06weighted} investigates the  optimal choice of the parameters $k_x$ that limits the false positive rate in expectation over both the input and the query distribution. Their approach follows the original Bloom filter analysis and casts the problem as an unconstrained optimization problem in which $k_x$ is allowed to be any real number (including negative). For more details, we refer the reader to Sec.~\ref{sec:related} This formulation and the fact that their false positive rate is taken as an average over $\calP_n$ leads to situations in which $k_x$ can be made arbitrarily large and, with high probability, the filter is filled with $1$s and has a high false positive probability (for instance, when a key is queried very rarely).  To avoid such situations, as we shall see next, optimal choices for $k_x$  exhibit some rather counter-intuitive trade-offs between $p_x$ and $q_x$.

	\subsection{Our Contributions}

	We start by discussing near-optimal choices for $k_x$ for a Weighted Bloom filter that is a  $(\calQ,\eps)$-filter for sets drawn from $\calP_n$. While this filter is not the most efficient of the filters we construct, reasoning through it helps us present our parametrizations and addresses the fact that Bloom filters remain well-liked in practice~\cite{DBLP:journals/comsur/LuoGMRL19}. Specifically,  we define $k_x$ as follows:\footnote{Throughout the paper, we employ $\log x$ to denote $\log_2 x$ and $\ln x$ to denote $\log_e x$.}
	
	\begin{align*}
		k_x&\triangleq
		\begin{cases}
			0 						& \text{if or $p_x > 1/n$ or $ q_x \leq \eps p_x $ } \; , \\
			\log(1/\eps \cdot q_x/p_x) & \text{if $\eps p_x < q_x \leq \min\!\set{p_x,\eps/n}$} \; , \\
			\log(1/\eps) 				& \text{if $q_x> p_x$ and $p_x \leq \eps/n $} \; , \\
			\log(1/(np_x)) 			& \text{if $q_x > \eps/n$ and $\eps/n < p_x \le 1/n$} \; .
		\end{cases}
	\end{align*}
	
	\begin{figure}[t]
		\centering
		\includegraphics[scale=0.5]{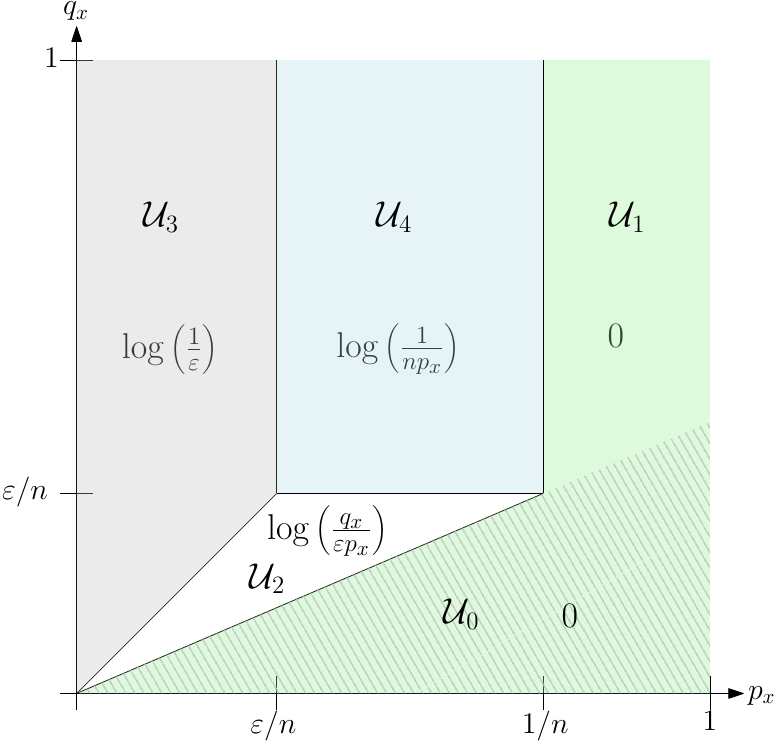}
		\caption{A schematic visualization of the different regimes for $k_x$.}
		\label{fig:domains}
	\end{figure}
	
	The first case covers the situation in which $x$ is very likely to be included in the set or is queried very rarely (relative to $p_x$). Intuitively, it makes sense in these cases to always say \text{YES} when queried.  Thus, we set $k_x=0$ and store no information about these keys. Conversely, the third case considers the case in which $x$ is queried so often (relative to $p_x$) that we need to explicitly keep the false positive probability below $\eps$, which is achieved by setting $k_x = \log(1/\eps)$. This is the largest number of hash functions we employ for any key, so in this sense, we are never worse than the classical Bloom filter. The second case interpolates smoothly between the first and third cases for elements that are rarely (but not very rarely) queried (compared to how likely they are to be inserted). Finally, the fourth case interpolates between the first and third case for elements that are not too rarely queried, in which case the precise query probability does not matter. See Fig.~\ref{fig:domains} for a visualization of these regimes.

	To further make sense of these regimes, we consider the case of uniform queries, i.e., $q_x = 1/u$, and assume that $\eps > n/u$, a standard assumption in filter design (otherwise, the filter would essentially have to answer correctly  on all queries and the lower bound of $n\log_2(1/\eps)-O(1)$ would not hold~\cite{carter1978exact,dietzfelbinger2008succinct}). Then in the two extremes, we would set $k_x=0$ for elements with $p_x \geq 1/(u\eps)$ (first case) and $k_x = \log(1/\eps)$ when $p_x \leq 1/u$ (third case). Keys  with $p_x$ between the two cases would exhibit the smooth interpolation $k_x = \log(1/(u\eps) \cdot 1/p_x)$, corresponding to the intuition that the more likely an element is to be inserted, the less information we should store about it (i.e., smaller $k_x$).

	\subparagraph*{The lower bound.} Given the above parameters, we then define the quantity  
	\begin{center}
		$\LB(\calP_n,\calQ, \eps) \triangleq \sum_{x \in \UU} p_x k_x$
	\end{center}  and show that, perhaps surprisingly, it gives a lower bound for the expected space that any $(\calQ,\eps)$-filter requires when the input set is drawn from $\calP_n$:
	
	\begin{theorem}[Lower bound - simplified]\label{thm:lbinf}
		Let $A$ be an algorithm and assume that for any input set $S \subseteq \UU$ with $\size{S} \le n$, $A(S)$ is a $(\calQ,\eps)$-filter for $S$. 
		Then the expected size of $A(S)$ must satisfy
		\begin{center}
			$\expectation{\calP_n, A}{\size{A(S)}} \geq \LB(\calP_n, \calQ, \eps) - 1 - 6n \; ,$
		\end{center}
		where $S$ is sampled with respect to $\calP_n$ and the queries are sampled with respect to $\calQ$.
	\end{theorem}

	Previous approaches for filter lower bounds show that there exists a set $S \subseteq U$ of size $n$ for which the filter needs to use $n\log_2(1/\eps)-O(1)$ bits~\cite{carter1978exact,dietzfelbinger2008succinct}.
	This type of  lower bound is still true in our model but it does not necessarily say anything meaningful, since the bad set $S$ could be sampled in $\calP_n$ with a negligible probability. Indeed, if we were to ignore the input distribution, then we would not be able to beat the worst input distribution and, in particular,  we would need to use at least $\sup_{\calP} \LB(\calP_n, \calQ, \eps) = \LB(\calQ_n, \calQ, \eps) = n \log(1/\eps)$ bits in expectation, where $\calQ_n$ denotes a distribution over $n$ independent draws from $\calQ$.
	
	In our model, it is therefore more natural to lower bound the \emph{expected} size of the filter over the randomness of the input set. Finally, we remark  that the full lower bound we prove is slightly stronger in that it holds for all but an unlikely collection of possible input sets, i.e. we only require that $A(S)$ is a $(\calQ,\eps)$-filter for $S \in \calT$ where $\calT \subseteq \mathbb{P}(U)$ and $\Prr{\calP_n}{S \not\in \calT} \le \frac{1}{\log u}$ (see Thm.~\ref{thm:lb}).
	
	\subparagraph*{The space-efficient filter.} We also show a filter design that asymptotically matches our space lower bound and executes operations in constant time in the worst case:
	
	\begin{theorem} [Space-efficient filter - simplified] Given $0<\eps<1$, there is a $(\calQ,\eps)$-filter with the following guarantees:
		\begin{itemize}
			\item it is a $(\calQ,\eps)$-filter with high probability over sets drawn from $\calP_n$, if $\sum_{x\in\UU} p_xq_x \leq \eps/n$,
			\item queries and insertions take constant time in the worst case,
			\item the space it requires is $(1+o_n(1))\cdot \LB(\calP_n,\calQ,\eps) + O(n)$ bits.
		\end{itemize}
	\end{theorem}
	
	The construction uses the $k_x$ values from above in conjunction with the fingerprinting technique of Carter et al.~\cite{carter1978exact} to obtain results that are comparable to state-of-the-art (classic) filter implementations that execute all operations (queries and insertions) in worst case constant time, and are space efficient, in the sense in which they require $(1+o_n(1))\cdot n\log(1/\eps) + O(n)$ bits~\cite{arbitman2010backyard, BerceaE20,BerceaE21, BenderFKKL22}.  The condition that  $\sum_{x\in\UU} p_xq_x \leq \eps/n$  can be seen as a generalization of the standard filter assumption that $\eps \geq n/u$.
	
	\subparagraph*{The Daisy Bloom filter.} For completeness, we also present our variant of the Weighted Bloom filter, which we call the \emph{Daisy Bloom filter}:\footnote{The daisy is one of our favorite flowers, especially when in full bloom, and is also a subsequence of ``\emph{d}yn\emph{a}m\emph{i}c \emph{s}trech\emph{y}'' which describes the key properties of our data structure. It is also the nickname of the Danish queen, whose residence is not far from the place where this work was conceived. Daisy Bloom filters are not related to any celebrities.}

	\begin{theorem} [Daisy Bloom filter - simplified] Given $0<\eps<1$, the Daisy Bloom filter has the following guarantees:
		\begin{itemize}
			\item it is a $(\calQ,\eps)$-filter with high probability over sets drawn from $\calP_n$, if $\sum_{x\in\UU} p_xq_x \leq \eps/n$,
			\item queries and insertions take at most $\ceil{\log_2(1/\eps)}$ time in the worst case,
			\item the space it requires is $\log(e)\cdot \LB(\calP_n,\calQ,\eps) + O(n)$ bits.
		\end{itemize}
	\end{theorem}
	
	In contrast to the weighted Bloom filters of Bruck et al.~\cite{BruckGJ06weighted}, the Daisy Bloom filter executes operations in time that is at most $\ceil{\log_2(1/\eps)}$ in the worst case (versus arbitrarily large) and achieves a false positive rate of at most $\eps$ with high probability over the input set (and not just on average).  We also depart in our analysis from their unconstrained optimization approach (to setting $k_x$ ) and instead use Bernstein's inequality to argue that, if the length of the array is set to  $\log(e)\cdot \LB(\calP_n,\calQ,\eps) + O(n)$ bits, then whp, at most half of the entries in the array will be set to $1$ (see Sec.~\ref{sec:ub} for more details).
	
	\subsection{Related Work}\label{sec:related}
	Filters have been studied extensively in the literature~\cite{bender2018bloom,BerceaE20,BerceaE21, carter1978exact, dietzfelbinger2008succinct, liu2020succinct, DBLP:conf/soda/PaghPR05, pagh2013approximate, porat2009optimal}, with Bloom filters perhaps the most widely employed variants in practice~\cite{DBLP:journals/comsur/LuoGMRL19}. 	Learning-based approaches to classic algorithm design have recently attracted a great deal of attention, see e.g.~\cite{ dong2020learning,galakatos2019fiting,hsu2019learning,kraska2018case,purohit2018improving}. For a comprehensive survey on learned data structures, we refer the reader to Ferragina and Vinciguerra~\cite{FerraginaV19}.
	
	\paragraph*{Weighted Bloom Filters} Given information about the probability of inserting and querying each element, Bruck, Gao and Jiang~\cite{BruckGJ06weighted} set out to find an optimal choice of the parameters $k_x$ that limit the false positive rate (in expectation over both the input and the query distribution).
	The approach is to solve an unconstrained optimization problem where the variables $k_x$ can be any real number.
	In a post-processing step each $k_x$ is rounded to the nearest non-negative integer.
	Unfortunately, this process does not lead to an optimal choice of parameters, and in fact, does not guarantee a non-trivial false positive rate.
	The issue is that the solution to the unconstrained problem may have many negative values of $k_x$, so even though the weighted sum $\sum_x p_x k_x$ is bounded, the post-processed sum $\sum_x p_x \max(k_x, 0)$ can be arbitrarily large.
	In particular, this is the case if at least one element is queried very rarely.
	This means that the weighted Bloom filter may consist only of 1s with high probability, resulting in a false positive probability of~1. 
	
	The above issue was noted by Wang, Ji, Dang, Zheng and Zhao~\cite{WangJDZZ15improved} who  attempt to correct the values for $k_x$, but their analysis still suffers from the same, more fundamental, problem: the existence of a very rare query element drives the false positive rate to 1.
	Wang et al.~\cite{WangJDZZ15improved} also show an information-theoretical ``approximate lower bound'' on the number of bits needed for a weighted Bloom filter with given distributions $\mathcal{P}$ and $\mathcal{Q}$.
	The sense in which the lower bound is approximate is not made precise, and the lower bound is certainly not tight (for example, it can be negative).

	\paragraph*{Partitioned Learned Bloom Filters} 
	There are several learned Bloom filter designs that assume that the filter has access to a learned model of the input set~\cite{kraska2018case, NEURIPS2018_0f49c89d,dai2020adaptive, vaidya2021partitioned}. The model is given a fixed input set $S$ and a representative sample of elements in $\UU\setminus S$ ( the query distribution is not specified).
	Given a query element $x$, the model returns a \emph{score} $s(x)\in [0,1]$, which can be intuitively thought of as the model's belief that $x\in S$.  Based on this score, Vaidya, Knorr, Mitzenmacher and Kraska~\cite{vaidya2021partitioned} choose a fixed number of $k$ thresholds,  partition the elements according to these thresholds, and build separate Bloom filters for each set of the partition. For fixed threshold values, they then formulate the optimization problem of setting the false positive rates $f_i$ such that the total space of the data structure is minimized and the overall false positive rate is at most a given $F$.
	
	As noted by Ferragina and Vinciguerra~\cite{FerraginaV19}, a significant drawback in these constructions is that the guarantees they provide depend significantly on the query set given as input to the machine learning component and in particular, the set being representative for the whole query distribution. We avoid this issue by making the dependencies on $q_x$ explicit and by bounding the average false positive probability even when just one element is queried. In addition, our data structure does not need to know the set $S$ in advance (and hence, training can be done just once, in a pre-processing phase), employs only one data structure, and our guarantees are robust to approximate values for $p_x$ and $q_x$.

	\subsection{Paper Organization}
	
	After some preliminaries, Sec.~\ref{sec:LB} shows our lower bound on the space usage. In Sec.~\ref{sec:opt-filter}, we discuss a space-efficient filter with constant time worst-case operations. Finally, Sec.~\ref{sec:ub} presents the analysis of the Daisy Bloom filter.

	\section{Preliminaries}

	For clarity, throughout the paper, we will distinguish between probabilities over the randomness of the input set, denoted by $\Prr{\calP_n}{\cdot}$, and probabilities over the internal randomness of the filter, denoted by $\Prr{A}{\cdot}$. 
	Joint probabilities are denoted by $\Prr{\calP_n,A}{\cdot}$. For the analysis, it will also make sense to partition the universe $\UU$ into the following $5$ parts:
	
	\begin{align*}
		\UU_0 &\triangleq \set {x\in\UU  \mid  q_x \leq \eps p_x} \; , \\
		\UU_1 &\triangleq \set {x\in\UU  \mid  q_x >  \eps p_x \text{ and } p_x > 1/n} \; , \\
		\UU_2 &\triangleq \set {x\in\UU  \mid   \eps p_x < q_x \leq \min\!\set{p_x,  \eps/n}} \; , \\ 
		\UU_3 &\triangleq \set {x\in\UU  \mid  q_x > p_x \text{ and }  \eps/n \ge p_x} \; , \\
		\UU_4 &\triangleq \set {x\in\UU  \mid  q_x >  \eps/n \text{ and }  \eps/n < p_x \le 1/n} \; .
	\end{align*}
	
	The high probability guarantees we obtain increase with $\LB(\calP_n,\calQ,  \eps)$. Therefore,
	such bounds are meaningful for distributions in which the optimal size $\LB(\calP_n,\calQ,  \eps)$  of a filter is not too small.
	Similarly, we can assume that the size of the universe is polynomial in $n$, and so $\log(1/ \eps) = O(\log n)$ in the standard case in which $ \eps>n/\size{\UU}$.
	Therefore, while in general $\LB(\calP_n,\calQ,  \eps)$ can be much smaller than $n\log_2(1/ \eps)$, we do require some mild dependency on $n$ for the high probability bounds to be meaningful. 	Finally, we recall the following classic result in data compression:


	\begin{theorem}[Kraft's inequality~\cite{thomas2006elements}]\label{thm:Kraft}
		For any instantaneous code (prefix code) over an alphabet of size $D$, the codeword lengths $\ell_1,\ell_2,\ldots, \ell_m$ must satisfy the inequality
		\begin{center}
			$ \sum_i D^{-\ell_i} \leq 1\;.$
		\end{center}
		Conversely, given a set of codeword lengths that satisfy this inequality,
		there exists an instantaneous code with these word lengths.
	\end{theorem}

	\section{The Lower Bound}\label{sec:LB}
	The goal of this section is to prove the lower bound from Thm.~\ref{thm:lbinf}.
	As discussed, we prove a slightly stronger statement where we allow our algorithm to not produce a $(\calQ,\eps)$-filter for some  input sets as long as the probability of sampling them is low. Formally, we show that:

	\begin{theorem}\label{thm:lb}
		Let $\calT \subseteq \mathbb{P}(\UU)$ be given such that $\Prr{\calP_n}{S \not\in \calT} \le \frac{1}{\log u}$.
		If $A$ is an algorithm such that for all $S \in \calT$, $A(S)$ is a $(\calQ, \eps)$-filter for $S$.
		Then the expected size of $A(S)$ must satisfy
		$$\expectation{\calP_n, A}{\size{A(S)}} \geq \LB(\calP_n, \calQ, \eps) - 1 - 6n \; ,$$
		where $S$ is sampled with respect to $\calP_n$.
	\end{theorem}
	
	\begin{proof}
		Each instance $\calI$ of the data structure corresponds to a subset $\UU_{\calI} \subset \UU$ on which the data structure answers \YES.
		We denote the number of bits needed by such an instance by $\size{\calI}$.
		For any set $S \in \calT$, we have that $\calI = A(S)$ satisfies that $S \subseteq \UU_{\calI}$ and
		\[
		\expectation{A}{\sum_{x \in \UU_{\calI} \setminus S} q_x} \le \eps \; .
		\]
		The goal is to prove that
		\[
		\expectation{\calP_n, A}{\size{A(S)}} \geq n \cdot \parentheses{\sum_{x\in\UU_2} p_x\log\parentheses{\frac{1}{\eps} \cdot \frac{q_x}{p_x}} + \sum_{x\in\UU_3} p_x\log\frac{1}{\eps} + \sum_{x\in\UU_4} p_x\log \frac{1}{n p_x} } - 1 - 6n \; .
		\]

		We will lower bound $\expectation{\calP_n, A}{\size{A(S)}}$ by using it to encode an ordered sequence of $n$ elements drawn according to $\calP_{n}$.
		Specifically, for any ordered sequence of $n$ elements $\hat{S} \in \UU^n$, we let $S \subseteq U$ be the set of distinct elements and let $\calI = A(S)$ as above.
		We first note that to encode $\hat{S} \sim \calP_n$, in expectation, we need at least the entropy number of bits, i.e.,
		\begin{align}\label{eq:entropy}
			n \sum_{x \in \UU} p_x \log \frac{1}{p_x} \; .		
		\end{align}
		
		Now our encoding using $\calI$ will depend on whether $S \in \calT$ or not.
		First, we will use 1 bit to describe whether $S \in \calT$ or not.
		For $(x_i)_{i \in [n]} \in \hat{S}$, we will denote $b_i$ to be the number bits to encode $x_i$.
		If $S \not\in \calT$ then for all $i \in [n]$ we encode $x_i$ using $b_i = \ceil{\log(1/p_{x_i})}$ bits.
		If $S \in \calT$ then for all $i \in [n]$ we encode $x_i$ depending on which subset if $\UU$ it belongs to:
		\begin{enumerate}
			\item If $x_i \in \UU_0 \cup \UU_1$, we encode $x_i$ using $b_i = \ceil{\log(4/p_{x_i})}$ bits.
			\item If $x_i \in \UU_2$, we encode $x_i$ using $b_i = \ceil{\log\left(4 \tfrac{\sum_{y \in \UU_{\calI} \cap \UU_2} q_y}{q_{x_i}} \right)}$ bits.
			\item If $x_i \in \UU_3$, we encode $x_i$ using $b_i = \ceil{\log\left(4 \tfrac{\sum_{y \in \UU_{\calI} \cap \UU_3} p_y}{p_{x_i}} \right)}$ bits.
			\item If $x_i \in \UU_4$, we encode $x_i$ using $b_i = \ceil{\log\left(4 \size{\UU_{\calI} \cap \UU_4} \right)}$ bits.
		\end{enumerate}
		It is clear from the construction that we satisfy the requirement for Thm.\ref{thm:Kraft} thus there exists such an encoding.
		Now we will bound the expectation of the size of this encoding:
		\begin{align*}
			\expectation{\calP_n, A}{\size{A(S)} + 1 + \sum_{i \in [n]} b_i}
			&= \expectation{\calP_n, A}{\size{A(S)}} + 1 + \expectation{\calP_n, A}{\sum_{i \in [n]} b_i} \; .
		\end{align*}
		We will write $\expectation{\calP_n, A}{\sum_{i \in [n]} b_i} = \expectation{\calP_n, A}{\indicator{S \in \calT}\sum_{i \in [n]} b_i} + \expectation{\calP_n, A}{\indicator{S \not\in \calT}\sum_{i \in [n]} b_i}$, and bound each term separately.
		
		We start by bounding $\expectation{\calP_n, A}{\indicator{S \not\in \calT}\sum_{i \in [n]} b_i}$.
		\begin{align*}
			\expectation{\calP_n, A}{\indicator{S \not\in \calT}\sum_{i \in [n]} b_i}
			&= \expectation{\calP_n}{\indicator{S \not\in \calT} \sum_{i \in [n]} \ceil{\log(1/p_{x_i})}}
			\\&\le \Prr{\calP_n}{S \not\in \calT}n + \expectation{\calP_n}{\indicator{S \not\in \calT} \log\left(\prod_{i \in [n]} 1/p_{x_i}\right)}
			\\&= \Prr{\calP_n}{S \not\in \calT}n + \sum_{\hat{s} \in \UU^n}\indicator{\hat{s} \in \calT} \Prr{\calP_n}{\hat{S} = \hat{s}}\log\frac{1}{\Prr{\calP_n}{\hat{S} = \hat{s}} } 
		\end{align*}
		Now using Jensen's inequality we get that
		\[
		\sum_{\hat{s} \in \UU^n}\indicator{\hat{s} \in \calT} \Prr{\calP_n}{\hat{S} = \hat{s}} \log\frac{1}{\Prr{\calP_n}{\hat{S} = \hat{s}} } \le \Prr{\calP_n}{S \not\in T} \log\left( \tfrac{1}{\Prr{\calP_n}{S \not\in T} u^n} \right) \; .
		\]
		Putting this together with the fact that $\Prr{\calP_n}{S \not\in T} \le \tfrac{1}{\log u}$, we get that,
		\begin{align*}
			\expectation{\calP_n, A}{\indicator{S \not\in \calT}\sum_{i \in [n]} b_i}
			\le 2n
			\; .
		\end{align*}
		
		Now we bound $\expectation{\calP_n, A}{\indicator{S \in \calT}\sum_{i \in [n]} b_i} = \sum_{i \in [n]} \expectation{\calP_n, A}{\indicator{S \in \calT}b_i}$.
		We will bound $\expectation{\calP_n, A}{\indicator{S \in \calT}b_i}$ depending on which subset of $\UU$ that $x_i$ belongs to.
		
		If $x_i \in \UU_0 \cup \UU_1$, then we have that $\expectation{\calP_n, A}{\indicator{S \in \calT}b_i} \le \ceil{\log(4/p_{x_i})} \le 3 + \log(1/p_{x_i})$.

		
		If $x_i \in \UU_2$, define $Z_2 = \UU_{\calI} \cap \UU_2$. Then 	$$\expectation{\calP_n, A}{\indicator{S \in \calT}b_i}  \le 3 + \expectation{\calP_n, A}{\indicator{S \in \calT}\log\left(\tfrac{\sum_{y \in Z_2} q_y}{q_{x_i}} \right)} \;.$$ We know that $\sum_{y \in S \cap \UU_2} q_y \le \eps$ since $q_y \le\eps/n$ for all $y \in \UU_2$ and $\size{S} \le n$. We also know that $\expectation{A}{\sum_{x \in Z_2 \setminus S} q_x} \le \eps$ for $S \in \calT$.
		Now using Jensen's inequality we get that 
		\begin{align*}
			\expectation{\calP_n, A}{\indicator{S \in \calT}\log\left(\tfrac{\sum_{y \in Z_2} q_y}{q_{x_i}} \right)}
			&= \expectation{\calP_n}{\indicator{S \in \calT}\expectation{A}{\log\left(\tfrac{\sum_{y \in Z_2} q_y}{q_{x_i}} \right)}}
			\\&\le \expectation{\calP_n}{\indicator{S \in \calT} \log\left(\tfrac{\expectation{A}{\sum_{y \in Z_2} q_y}}{q_{x_i}} \right) }
			\\&\le \expectation{\calP_n}{\indicator{S \in \calT} \log\left(\tfrac{2\eps}{q_{x_i}} \right) } \le 1 + \expectation{\calP_n}{\log\left(\tfrac{\eps}{q_{x_i}} \right) } \; .
		\end{align*}


		If $x_i \in \UU_3$, define $Z_3=\UU_{\calI} \cap \UU_3$. Then $$\expectation{\calP_n, A}{\indicator{S \in \calT}b_i}  \le 3 + \expectation{\calP_n, A}{\indicator{S \in \calT}\log\left(\tfrac{\sum_{y \in Z_3} p_y}{p_{x_i}} \right)} \;.$$ We know that $\sum_{y \in S \cap \UU_3} p_y \le \eps$ since $p_y \le \eps/n$ for all $y \in \UU_3$ and $\size{S} \le n$.
		We also know that $\expectation{A}{\sum_{x \in Z_3 \setminus S} p_x} \le \expectation{A}{\sum_{x \in Z_3\setminus S} q_x} \le \eps$ for $S \in \calT$.
		Using Jensen's inequality we get that,
		\begin{align*}
			\expectation{\calP_n, A}{\indicator{S \in \calT}\log\left(\tfrac{\sum_{y \in Z_3} p_y}{p_{x_i}} \right)}
			&= \expectation{\calP_n}{\indicator{S \in \calT}\expectation{A}{\log\left(\tfrac{\sum_{y \in Z_3} p_y}{p_{x_i}} \right)}}
			\\&\le \expectation{\calP_n}{\indicator{S \in \calT} \log\left(\tfrac{\expectation{A}{\sum_{y \in Z_3} p_y}}{p_{x_i}} \right) }
			\\&\le \expectation{\calP_n}{\indicator{S \in \calT} \log\left(\tfrac{2\eps}{p_{x_i}} \right) } \le 1 +  \expectation{\calP_n}{\log\left(\tfrac{\eps}{p_{x_i}} \right) }\; .
		\end{align*}


		If $x_i \in \UU_4$, define $Z_4 = \UU_{\calI} \cap \UU_4$. Then $\expectation{\calP_n, A}{\indicator{S \in \calT}b_i}  \le 3 + \expectation{\calP_n, A}{\indicator{S \in \calT}\log\left(\size{Z_4} \right))} $. We know that $\size{Z_4} = \size{Z_4\cap S} + \size{Z_4 \setminus S} \le n + \tfrac{n}{\eps} \sum_{x \in Z_4 \setminus S} q_x$ since $q_y > \eps/n$ for all $y \in \UU_4$ and $\size{S} \le n$.
		Using Jensen's inequality we get that, for $Z'_4 = Z_4 \setminus S$:
		
		\begin{align*}
			\expectation{\calP_n, A}{\indicator{S \in \calT}\log\left(n + \frac{n}{\eps} \sum_{x \in Z'_4} q_x \right)}
			&= \expectation{\calP_n}{\indicator{S \in \calT}\expectation{A}{\log\left(n + \frac{n}{\eps}  \sum_{x \in Z'_4} q_x \right)}}
			\\&\le \expectation{\calP_n}{\indicator{S \in \calT} \log\left(\expectation{A}{n + \frac{n}{\eps} \sum_{x \in Z'_4} q_x} \right) }
			\\&\le \expectation{\calP_n}{\indicator{S \in \calT} \log\left( 2n \right) } \le 1 + \log(n) \; .
		\end{align*}

		Combining it all we get an encoding that in expectation uses at most 
		\begin{align*}
			&\expectation{\calP_n, A}{\size{A(S)}} + 1 + 6n + 
			\\&\qquad\qquad\sum_{x\in (\UU_0\cup \UU_1)} p_x \log(1/p_x) + \sum_{x\in \UU_2} p_x \log(\eps/q_x) + \sum_{x\in \UU_3} p_x \log(\eps/p_x) + \sum_{x\in \UU_4} p_x \log n
			\; .
		\end{align*}
		bits to encode $\hat{S}$.
		Comparing this with Eq.\ref{eq:entropy} we get that,
		\begin{align*}
			\expectation{\calP_n, A}{\size{A(S)}}
			\ge n \cdot \parentheses{\sum_{x\in\UU_2} p_x\log\parentheses{\frac{1}{\eps} \cdot \frac{q_x}{p_x}} + \sum_{x\in\UU_3} p_x\log\frac{1}{\eps} + \sum_{x\in\UU_4} p_x\frac{1}{n p_x} } - 1 - 6n
			\; .
		\end{align*}
		This proves the claim. 
	\end{proof}

	\section{Space-Efficient  Filter}\label{sec:opt-filter}
	
	In this section, we show how one can use the $k_x$ values proposed to design a space-efficient $\calQ$-filter with worst-case constant time operations. Formally, we show that:
	
	\begin{theorem}\label{thm:spaceefficient} Assume that $\calP_n$ and $\calQ$ satisfy $n\sum_{x\in\UU}p_xq_x \leq \eps$. Then there exists a filter with the following guarantees:
		\begin{itemize}
			\item there exists  $\calT \subseteq \mathbb{P}(U)$ where a set $S \in \calT$ with high probability over the randomness of $\calP_n$, such that the filter is a $(\calQ,\eps)$-filter for any $S \in \calT$,
			\item the filter uses
			$(1+o_n(1)) \cdot LB(\calP_n,\calQ,\eps) + O(n) $ bits,
			\item the filter executes queries and insertions in worst case constant time and,
			\item the filter does not fail with high probability over its internal randomness.
		\end{itemize}	
	\end{theorem}

	\subsection{Construction}
	For $j\in \set{1,\ldots,\ceil{\log(1/\eps)}}$, we let $\UU^{(j)} \triangleq \set{x\in \UU \mid \ceil{k_x} = j }$ denote the set of elements that hash to $j$ locations in the Daisy Bloom filter and $P_j \triangleq \sum_{x\in\UU^{(j)} } p_x$ denote the probability that we select an element from $\UU^{(j)}$ in one sample from $\calP$. Then $n_j \triangleq n\cdot P_j$ denotes the average number of elements from  $\UU^{(j)}$ that we expect to see in the input set. 
	We distinguish between the sets $\set{\UU^{(j)}}$ depending on their corresponding $n_j$.
	Specifically, we say that  $\UU^{(j)}$ is a \emph{rare class} if $n_j < n/\log^c n$, for some constant $c>2$, and otherwise we say that $\UU^{(j)}$ is a \emph{frequent class}. We further define  $\UU_r\subseteq \UU$ to be the set of all elements that are in a rare class, i.e.,  $\UU^{(j)} \subseteq \UU$ if and only if  $\UU^{(j)}$  is a rare class.

	Now let $\calF(\eps,n)$ be a (standard) filter implementation for at most $n$ elements with false positive probability at most $\eps$. We focus on implementations that execute all operations (queries and insertions) in worst-case constant time, and are space efficient: they require $(1+o_n(1))\cdot n\log(1/\eps) + O(n)$ bits~\cite{arbitman2010backyard, BerceaE20,BerceaE21, BenderFKKL22}. 
	We employ $\ceil{\log(1/\eps)}+1$ instantiations of $\calF$, which we denote by  $\calF_1,\ldots,\calF_{\ceil{\log(1/\eps)}}$ and $\calF_r$. They are parametrized as follows: for $j\in \set{1,\ldots,\ceil{\log(1/\eps)}}$, we further define $N_ j \triangleq (1+1/\log n)\cdot n_j$ and  instantiate $\calF_j = \calF(2^{-j}, N_j)$. 
	We instantiate $\calF_r$ as $\calF_r = \calF(F, N_r)$, where $N_r \triangleq \Theta(n/\log^{c-1} n)$.
	
	\subparagraph*{Operations.} We distinguish between elements that are in a frequent class and elements that are in a rare class. If an element is in a frequent class $\UU^{(j)}$, then operations are forwarded to the corresponding filter $\calF_j$. Otherwise, the operation is forwarded to $\calF_r$. Since all the filters we employ perform operations in constant time in the worst case, the same holds for our construction\footnote{We note here that it is possible to combine the filters $\calF_1,\ldots,\calF_{\ceil{\log(1/\eps)}}$ into one single data structure. One could use, for example, the balls-into-bins implementation in~\cite{BerceaE20}, where elements are randomly assigned to one of $n/\Theta(\log n/ (\log(1/\eps)))$ buckets and the buckets explicitly store random strings of length $\log(1/\eps)$ associated with the elements that hash into them. Combining $\calF_1,\ldots,\calF_{\ceil{\log(1/\eps)}}$ would then entail ``superimposing'' their buckets.}.
	
	\subparagraph*{Space.} We now bound the total number of bits that $\calF_1,\ldots,\calF_{\ceil{\log(1/\eps)}}$ and $\calF_r$ require:
	
	\begin{lemma}
		The above filter requires $(1+o_n(1))\cdot 	\LB(\calP_n,\calQ, \eps)  + O(n)$ bits.
	\end{lemma}
	\begin{proof} Recall that $\calF_1,\ldots,\calF_{\ceil{\log(1/\eps)}}$ and $\calF_r$ are instantiations of a filter $\calF$ which requires $(1+f(n))\cdot n\log(1/\eps) + O(n)$ bits for a set of $n$ elements and false positive probability $\eps$, where the function $f(n)$ satisfies $f(n)=o_n(1)$~\cite{arbitman2010backyard,BerceaE20,BerceaE21,BenderFKKL22}. For simplicity, we choose the implementation in~\cite{BerceaE20}, where $f(n) = \Theta(\log\log n / \sqrt{\log n})$. Consequently, for $j\in \set{1,\ldots,\ceil{\log(1/\eps)}}$, the space of $\calF_j$ is:
		$$ (1+ f(N_j)) \cdot N_j\log(1/2^{-j}) + O(N_j)  =  (1+ f(N_j)) \cdot N_j\cdot j + O(N_j)$$
		bits. Since $\calF_j$ is instantiated only for frequent classes, it follows that $n\geq N_j \geq  n/\log^c n$ and hence, $f(N_j) = \Theta(f(n))$ for all $j$ with $\UU^{(j)}$ a frequent class. Furthermore, by definition, we have that 
		$j \leq k_x + 1$ for all $x\in\UU^{(j)}$ and $N_j = (1+1/\log n) \cdot n_j = (1+1/\log n) \cdot n \sum_{x\in \UU^{(j)}} p_x $. Therefore,
		the space that $\calF_j$ requires can be upper bounded by
		$$ (1+ \Theta(f(n))) \cdot n \sum_{x\in \UU^{(j)}} p_x k_x + O(N_j) \;.$$
		Since  $ \sum_j N_j = (1+1/\log n) \cdot n$ 
		we get that, in the worst case in which all the classes are frequent, the filters $\calF_1,\ldots,\calF_{\ceil{\log(1/\eps)}}$  require
		$$ (1+ o_n(1)) \cdot n  \sum_{x\in \UU} p_x k_x + O(n) = (1+o_n(1))\cdot 	\LB(\calP_n,\calQ, \eps)  + O(n)$$
		bits. The space of the final filter $\calF_r$  is upper bounded by 
		$\Theta(n/\log^{c-1}n) \cdot \log(1/\eps) = \Theta(n/\log^{c-2}n) = o(n)$ bits for any constant $c>2$. The claim follows.
	\end{proof}
	
	\subsection{Analysis}
	
	In this section, we show that the filter described above does not fail whp and that it achieves a false positive probability of at most $3\eps$ with respect to $\calQ$.  In our construction, there are two sources of failure: when the number of elements which we insert into each filter exceeds the maximum capacity of the filter, and when the filters themselves fail as a consequence of their internal randomness. 
	In the latter case, we note that the failure probability of $\calF(\eps,n)$ is guaranteed to be at most $1/\poly(n)$, where the degree of the polynomial is a constant of our choosing~\cite{arbitman2010backyard, BerceaE20,BerceaE21, BenderFKKL22}.
	Since all of the instantiations we employ have maximum capacities which are $\Omega(n/\polylog(n))$, we conclude that each of these separate instantiations also fails with probability at most $1/\poly(n)$, and therefore, by a union bound over the $\ceil{\log(1/\eps)}+1 = O(\log n)$ instantiations, we get that some filter fails with probability at most $1/\poly(n)$. We now show that the maximum capacities we set for each filter suffice.
	
	\begin{lemma}
		Whp, at most $N_r = \Theta(n/\log^{c-1} n)$ elements from $\UU_r$  are sampled in the input set.
	\end{lemma}
	\begin{proof}
		Let $\UU^{(j)}$ be a rare class and let $X_j$ denote the number of elements from $\UU^{(j)}$ that we sample in the input set.
		By definition, the  expected number of elements we see from  $\UU^{(j)}$ satisfies $\expectation{\calP}{X_j} = n_j < n/\log^c n$, for some constant $c>2$. By Chernoff bound, we therefore get that:
		\begin{align*}
			\Pr{X_j >  6n/\log^{c} n } &\leq 2^{-6n/\log^{c} n } \;.
		\end{align*}
		There are at most $\ceil{\log(1/\eps)} = O(\log n)$ possible rare classes, and so, by the union bound, the number of elements from $\UU_r$ that we sample in the input set is at most   $N_r= \Theta(n/\log^{c-1} n  )$ whp.
	\end{proof}

	We now focus on sampling elements from a frequent class:
	\begin{lemma} Let $\UU^{(j)}$ be a frequent class. Then,  whp, at most $N_j$ elements from $\UU^{(j)}$ are sampled in the input set.
	\end{lemma}
	\begin{proof}
		Let $X_j$ denote the number of elements from $\UU^{(j)}$ that we sample in the input set and note that $\expectation{\calP}{X_j} = n_j \geq n/\log^c n$. By Chernoff:
		\begin{align*}
			\Pr{X_j >  N_j} = \Pr{X_j >  (1+1/\log n) \cdot n_j} &\leq \exp(-\Theta(n/\log^{c-2} n)) \leq 1/\poly n \;.
		\end{align*}
		This concludes our proof.
	\end{proof}

	Finally, we bound the false positive rate of the  filter:
	
	\begin{lemma}\label{lem:fpr}  Assume that $\calP_n$ and $\calQ$ satisfy $n\sum_{x\in\UU}p_xq_x \leq \eps$ and that the input set $S$ does not make $\calF_1,\ldots, \calF_{\ceil{\log(1/\eps)}}$ and $\calF_r$ fail. Then the filter described is a $(\calQ,3\eps)$-filter on $S$.
	\end{lemma}
	\begin{proof}
		Fix an input set $S$ and denote by $A'(S,x)$ the output of the filter when queried for an element $x$. We are interested in bounding
		$\Pr{A'(S,x) = \textsf{YES}}$ for an element $x\notin S$. 
		If $x\in \UU_r$, then we forward the query operation to $\calF_r$, which guarantees that $\Pr{A'(S,x) = \textsf{YES}} \leq \eps$. Therefore:
		$$\sum_{x\in\UU_r } q_x \cdot  \Pr{A'(S,x) = \textsf{YES}}  \leq \sum_{x\in\UU_r } q_x \cdot \eps \;,$$

		Otherwise, if $x\notin\UU_r$, the query is forwarded to $\calF_j$, where $j = \ceil{k_x}$. In this case,
		$\Pr{A'(S,x) = \textsf{YES}} \leq 2^{-j} \leq 2^{-k_x}$ and we get that 
		$$\sum_{x\in\UU_0 \cup \UU_2 \setminus \UU_r} q_x \cdot  \Pr{A'(S,x) = \textsf{YES}}  \leq  \sum_{x\in\UU_0  \setminus \UU_r} q_x  + \sum_{x\in\UU_2 \setminus \UU_r} p_x \eps \le  \sum_{x\in\UU_0 \cup \UU_2 \setminus \UU_r} p_x \cdot \eps\;,$$
		and similarly, 
		$$\sum_{x\in\UU_1 \cup \UU_4 \setminus \UU_r} q_x \cdot  \Pr{A'(S,x)  = \textsf{YES}} \le \sum_{x\in\UU_1  \setminus \UU_r} q_x + \sum_{x\in\UU_4  \setminus \UU_r} n p_xq_x   \leq \sum_{x\in\UU_1 \cup \UU_4 \setminus \UU_r} n p_x q_x \;.$$
		Finally, we have that
		$$\sum_{x\in\UU_3 \setminus \UU_r} q_x \cdot  \Pr{A'(S,x) = \textsf{YES}}  \leq \sum_{x\in\UU_3 \setminus \UU_r} q_x \cdot \eps \;.$$
		Adding all of these quantities, we obtain the claim.
	\end{proof}
	
	\subsection{Remarks}\label{sec:estimate}
	The filter construction assumes that we know, in advance, whether a class $\UU^{(j)}$ is frequent and, if so, what is the value of its corresponding $P_j = \sum_{x\in \UU^{(j)}} p_x$. This is because we employ fixed capacity filters  which require us to provide an upper bound on the cardinality of the input set $S \cap \UU^{(j)}$ at all points in time. 
	We note that this assumption can be alleviated in two ways: on one hand, one can employ filters that do not require us to know the size of the input set in advance~\cite{pagh2013approximate, BerceaE22}. This would incur an additional $\Theta(n\log\log n)$ bits in the space consumption of our filter (operations would remain constant time worst case).
	
	On the other hand, one can estimate $P_j$ for all frequent classes $\UU^{(j)}$ if we are allowed to take  $\polylog(n)$ samples from $\calP$ before constructing the filter. Specifically, fix $j \in \set{1,\ldots,\ceil{\log(1/\eps)}}$ and take $\ell = O(\log^{2c} n)$ samples from $\calP$. Define $Z_j$ to be the number of sampled elements that are in $\UU^{(j)}$. Then, if $\UU^{(j)}$ is indeed frequent, by the standard Chernoff bound we get that, whp, $Z_j > \log^c n $ elements and $Z_j/\ell$ is an unbiased estimator for $P_j$ with the guarantee that
	$P_j = (1\pm O(1/\log^{(c-1)/2} n))\cdot Z_j/\ell $ whp. Note that we can tolerate such an estimate since we set the maximum capacity of each filter to be $N_j = (1+1/\log n)\cdot nP_j$. A similar argument can be used for estimating the lower bound $\LB(\calP_n,\calQ, \eps)  = n\cdot \sum_{x\in\UU} p_x \cdot k_x$ whp. Specifically, by the definition of $P_j$, we have that
	$$ \sum_{x\in\UU} p_x \cdot \ceil{k_x} = \sum_{j=1}^{\ceil{\log(1/\eps)}} j\cdot P_j \;.$$
	If $\UU^{(j)}$ is a frequent class, then by the above argument we have an estimate of its $P_j$. Otherwise, we know that $P_j < 1/\log^c n$ and, since $j\leq \ceil{\log(1/\eps)} =O(\log n)$, get that
	$$\sum_{j \text{ s.t. } \UU^{(j)} \in \UU_r}j\cdot P_j \leq \sum_{j \text{ s.t. } \UU^{(j)} \in \UU_r}\ceil{\log(1/\eps)} \cdot 1/\log^c n\leq 1/\log^{c-2} n \;,$$
	which contributes a $o(n)$ term to the lower bound.

	\section{The Daisy Bloom Filter Analysis}\label{sec:ub}
	In this section, we analyse the behaviour of the Daisy Bloom filter with the values $k_x$ denoting the number of hash functions that we use to hash $x$ into the array.
	Let $X_i$ denote the number of hash functions that are employed when we sample in the $i^{th}$ round, i.e., $X_ i = k_x$ with probability $p_x$. 
	Then $X = \sum_i X_i$ denotes the number of locations that are set in the Bloom filter (where the same location might be set multiple times).
	Moreover, $\expectation{\calP_n}{X} = n \cdot  \sum_{x\in \UU} p_x k_x = \LB(\calP_n,\calQ,\eps)$, since the $\set{X_i}_i$ are identically distributed. 
	We then set the length $m$  of the Daisy Bloom filter array  to $$m\triangleq  \expectation{\calP_n}{X}/ \ln2 \;.$$

	The remainder of this section is dedicated to proving the following statement:
	
	\begin{theorem}\label{ub}
		Assume that $\calP_n$ and $\calQ$ satisfy $n\sum_{x\in\UU}p_xq_x \leq \eps$. Then there exists $\calT \subseteq \mathbb{P}(U)$ such that $S \in \calT$ with high probability over the randomness of $\calP_n$, and for all $S \in \calT$ the Daisy Bloom Filter is a $(\calQ,\eps)$-filter for $S$.
		The Daisy Bloom filter uses $ \log(e)\cdot \LB(\calP_n,\calQ,\eps) + O(n)$ bits and executes all operations in at most $\ceil{\log(1/\eps)}$ time in the worst case.
	\end{theorem}

	The sets in $\calT$ are the sets for which $X \approx \expectation{\calP_n}{X} = \LB(\calP_n, \calQ, \eps)$.
	The reason we constrain ourselves to these sets, is that if $X \gg \LB(\calP_n, \calQ,\eps)$ then most bits will be set to 1 which will make the false positive rate large.
	We will bound the probability that $X \gg \LB(\calP_n, \calQ,\eps)$ by using Bernstein's inequality and here, the following observation becomes crucial:
	
	\begin{observation}\label{obs:upperbound} For every $x\in\UU$, $k_x \leq \log(1/\eps)$.
	\end{observation}
	\begin{proof}
		For $x \in \UU_0 \cup \UU_1$, we have that $k_x = 0$ which is clearly less than $\log(1/\eps)$.
		For $x \in \UU_3$ we have that $k_x = \log(1/\eps)$ and again the statement holds trivially.
		For $x\in \UU_2$, we have that $q_x\leq p_x$ and so $k_x = \log(1/\eps \cdot q_x/p_x) \leq \log(1/\eps)$.
		For $x\in\UU_4$, we have that $p_x > F/n$ and so $k_x = \log(1/(np_x))<\log(1/\eps)$.
	\end{proof}

	We are now ready to prove that the random variable $X$ is concentrated around its expectation.
	\begin{lemma}\label{claim:Xconcentrated} For any $\delta>0$,
		\begin{align*}
			\Prr{\calP_n}{X   > (1+\tau)\cdot \expectation{\calP_n}{X}} & \leq   \exp\parentheses{-\frac{\tau^2\ln 2}{2(1+\tau/3)}\cdot \frac{m}{\log(1/\eps)} } 
		\end{align*}
	\end{lemma}
	\begin{proof}
		The random variables $\set{X_i}_i$ are independent and 
		$X_i \leq b \triangleq \log(1/\eps)$ for all $i$ by Obs.~\ref{obs:upperbound}. 
		We apply Bernstein's inequality~\cite{dubashipanconesi}:
		$$ \Prr{\calP_n}{X - \expectation{\calP_n}{X} > t} \leq \exp\parentheses{-\frac{t^2/2}{n \variance{\calP_n}{X_1} + bt/3} }\;.$$
		Note that
		$\variance{\calP_n}{X_i} \leq \expectation{\calP_n}{X_i^2} \leq b\cdot \expectation{\calP_n}{X_i}$.
		Setting $t= \tau \cdot  \expectation{\calP_n}{X} = \tau n \cdot \expectation{\calP_n}{X_1}$, we get that
		
		\begin{align*}
			\Prr{\calP_n}{X   >(1+\tau)\expectation{\calP_n}{X}} &\leq \exp\parentheses{-\frac{\tau^2}{2}\cdot \frac{n^2(\expectation{\calP_n}{X_1})^2}{nb\cdot \expectation{\calP_n}{X_1}+ \tau/3 \cdot n b \cdot \expectation{\calP_n}{X_1}} }\\
			& = \exp\parentheses{-\frac{\tau^2}{2(1+\tau/3)}\cdot \frac{n\expectation{\calP_n}{X_1}}{b} }   \;.
		\end{align*}
		The claim follows by noticing that $n\expectation{\calP_n}{X_1} = m\ln 2$.
	\end{proof}

	We can now prove that as long as $X \le (1 + 1/(2\log(1/\eps))) \cdot \expectation{\calP_n}{X} $, the Daisy Bloom filter is a $(\calQ,6 \eps)$-filter for $S$. We consider the fraction $\rho$ of entries in the array that are set to $0$ after we have inserted the elements of the set. We then show that $\rho$ is close to $1/2$ with high probability over the input set and the randomness of the hash functions. Conditioned on this, we then have that the probability that we make a mistake for $x$ is at most $2^{-k_x+1}$. The false positive rate is then derived similarly to that of Lemma~\ref{lem:fpr}.

	\begin{lemma}\label{boundfprate}
		Assume that $\calP_n$ and $\calQ$ satisfy $n\sum_{x\in\UU}p_xq_x \leq \eps$, and that $X \le (1 + 1/(2\log(1/\eps))) \expectation{\calP_n}{X}$.
		Then, whp, the Daisy Bloom filter is a $(\calQ,6\eps)$-filter on $S$.
	\end{lemma}
	
	\begin{proof}
		Let $\rho \in [0,1]$ denote the fraction of entries in the Daisy Bloom filter that are set to $0$ after we have inserted the elements of the set.
		Recall that the random variable $X$ denotes the total number of entries that are set in the Bloom filter, including multiplicities.
		In the worst case, all the entries to the filter are distinct, and we have $X$ independent chances to set a specific bit to $1$. Therefore
		$$ \expectation{h}{\rho | X} \geq \parentheses{1-\frac{1}{m}}^X \approx e^{-X/m} = 2^{-X/\expectation{\calP_n}{X}}\;.$$
		
		Moreover, by applying a Chernoff bound for negatively associated random variables, we have that for any $0<\gamma<1$, 
		
		\begin{align}\label{conc}\Prr{A}{\rho \leq (1-\gamma)\cdot \parentheses{1-\frac{1}{m}}^X \,\middle\vert\, X} & \leq \exp\parentheses{-m \parentheses{1-\frac{1}{m}}^X \cdot \gamma^2/2}
		\end{align}
		
		We now let $B_\delta$ denote the event that $\parentheses{1-\frac{1}{m}}^X > (1-\delta)\cdot \frac{1}{2}$ and $B_{\gamma}$ the event that
		$\rho > (1-\gamma) \parentheses{1-\frac{1}{m}}^X $.
		We then choose  $0<\delta<1$ and $0<\gamma<1$ such that $B_\delta$ and $B_{\delta'}$ imply that 
		$$ \rho \geq 1- 2^{1/\log(1/\eps)}\cdot \frac{1}{2}\;.$$
		Moreover, for our choices of $\delta$ and $\gamma$, we have that both $B_\delta$ and $B_\gamma | B_\delta$ occur with high probability.\footnote{We implicitly assume here that $2^{1/\log(1/\eps)}\cdot \leq 2$, i.e., $\eps\leq 1/2$. Notice that this does not affect the overall result, since the false positive rate we obtain is $5\cdot \eps$.}
		We refer the reader to App.\ref{app:deltagamma} for $\delta$ and $\gamma$.
		Conditioned on $B_\delta$ and $B_\gamma$, we get that, for an $x\notin S$, since $k_x \leq \log(1/\eps)$,
		$$\Prr{A}{A(S,x) = \textsf{YES} | B_\delta \wedge B_\gamma} =(1-\rho)^{k_x}\leq 2^{k_x/b}\cdot 2^{-k_x} \leq 2\cdot 2^{-k_x} $$
		
		We bound the false positive rate on each partition.
		For $x\in \UU_0$, i.e.,  with $q_x\leq Fp_x$ and $k_x = 0$, we can upper bound the false positive rate  as such
		\begin{align*}
			\sum_{x\in \UU_0} q_x \cdot \Pr{A(S,x) = \textsf{YES}| B_\delta \wedge B_\gamma} &\leq  \sum_{x\in \UU_0} q_x \leq  \sum_{x\in \UU_0} p_x \cdot \eps \;.
		\end{align*} 
		
		For $x\in\UU_1$ with $p_x>1/n$ and $k_x=0$, we have  the following
		\begin{align*}
			\sum_{x\in \UU_1} q_x \cdot \Pr{A(S,x) = \textsf{YES}| B_\delta \wedge B_\gamma} &\leq  \sum_{x\in \UU_1} q_x < n\sum_{x\in \UU_1} p_xq_x \;.
		\end{align*} 
		
		For $x\in\UU_2$ with $k_x = \log(1/\eps \cdot q_x/p_x)$, we have the following
		\begin{align*}
			\sum_{x\in \UU_2} q_x \cdot \Pr{A(S,x) = \textsf{YES}   | B_\delta \wedge B_\gamma} &\leq  \sum_{x\in \UU_2} q_x \cdot 2\cdot 2^{-k_x} =  \sum_{x\in \UU_2} q_x \cdot 2\cdot p_x/q_x \cdot \eps\\
			& =  \sum_{x\in \UU_2} p_x \cdot 2\eps \;.
		\end{align*} 
		
		For $x\in\UU_3$ with $k_x = \log(1/\eps)$,
		\begin{align*}
			\sum_{x\in \UU_3} q_x \cdot \Pr{A(S,x) = \textsf{YES}   | B_\delta \wedge B_\gamma} &\leq  \sum_{x\in \UU_3} q_x \cdot 2\eps  \leq 2\eps\;.
		\end{align*} 
		
		For $x\in\UU_4$ with $k_x = \log(1/(np_x))$,
		\begin{align*}
			\sum_{x\in \UU_4} q_x \cdot \Pr{A(S,x) = \textsf{YES}   | B_\delta \wedge B_\gamma} &\leq  \sum_{x\in \UU_4} q_x \cdot 2 np_x  =  2n \sum_{x\in \UU_4} p_x q_x \;.
		\end{align*} 
		
		For the overall false positive rate, note that the total false positive rate in $\UU_0$ and $\UU_2$ is at most $2\eps$. Similarly for the false positive rate in $\UU_3$.  For the remaining partitions $\UU_1$ and $\UU_4$, we have that it is at most
		$$ 2n\sum_{x\in \UU_1 \cup \UU_4} p_xq_x\;.$$
		From our assumption, this later term is at most $2\eps$ as well.
		
	\end{proof}
	Combining the above with Lemma~\ref{claim:Xconcentrated} we get that with probability $1 - \exp\parentheses{-\frac{m}{\Theta(\log^3(1/\eps))} }$ over the randomness of the input set, the Daisy Bloom filter is a $(\calQ,6\eps)$-filter for $S$.
	This is exactly the statement of Thm.~\ref{ub}.

	\bibliography{./references}
	\appendix

	\section{Choices for \texorpdfstring{$\delta$ and $\gamma$}{delta and gamma} in the proof of Lemma~\texorpdfstring{\ref{boundfprate}}{Lemma \ref{boundfprate}} }\label{app:deltagamma} 
	In the proof of Lemma~\ref{boundfprate}, we define the event
	$B_\delta$ to be when $ \parentheses{1-\frac{1}{m}}^X>(1-\delta)\cdot \frac{1}{2}$ and the event $B_{\gamma}$  to be when
	$\rho > (1-\gamma) \parentheses{1-\frac{1}{m}}^X $.
	We then claim that we can choose values $0<\delta<1$ and $0<\gamma<1$ such that $B_\delta$ and $B_{\delta'}$ occur with high probability and imply that 
	$$ \rho \geq 1- 2^{1/\log(1/\eps)}\cdot \frac{1}{2}\;.$$
	In this section, we show such possible values for $\delta$ and $\gamma$ and bound the probability of $B_{\delta}$ and $B_\gamma|B_\delta$.
	Specifically, we let $b = \log(1/\eps)$ and define $\delta  \triangleq 2^{1/(2b)}-1$ and $\gamma\triangleq (1-2^{-1/(2b)})\cdot 2^{1/b-1}$. Then $B_\delta \wedge B_\gamma$ imply that
	\begin{align*}
		\rho &> (1-\gamma)\cdot (1-\delta) \frac{1}{2} = (1-\gamma)\cdot \parentheses{2-2^{1/(2b)}}\frac{1}{2} = (1-\gamma)\cdot \parentheses{1-2^{1/(2b)-1}}\\
		& = (1-\gamma)\cdot \parentheses{1-2^{-1/(2b)} \cdot 2^{1/b-1}} \\
		&\geq 1-(\gamma/2^{1/b-1} + 2^{-1/(2b)} ) \cdot 2^{1/b-1} = 1- 2^{1/b-1} \;.
	\end{align*}
	We now show that $B_\delta$ and $B_\gamma| B_\delta$ occur with high probability. 
	
	Let $\tau= 1/(2b)=1/(2\log(1/\eps))$ and note that if $X   \leq(1+\tau)\cdot \expectation{\calP_n}{X}$, we get that for $m\geq (2b)^4$:
	\begin{align*}
		\parentheses{1-\frac{1}{m}}^X&\geq \parentheses{\parentheses{1-\frac{1}{m}} \cdot \frac{1}{e}}^{X/m} \geq \parentheses{1-\frac{1}{m}}^{(1+\tau)\ln 2}\cdot 2^{-1-\tau}\\
		&\geq \parentheses{1-\frac{1}{m}}^{2}\cdot 2^{-1-\tau} \geq \parentheses{1-\frac{2}{m}}\cdot 2^{-1-\tau} \geq \parentheses{1-\frac{2}{(2b)^4}}\cdot 2^{-1-\tau} \\
		&\geq (1-\delta)\cdot 2^{\tau} \cdot 2^{-1-\tau} = (1-\delta)\frac{1}{2}\;,
	\end{align*}
	where the inequality $ 1-2/(2b)^4 \geq (1-\delta)\cdot 2^{\tau}  = (2-2^{1/(2b)})\cdot 2^{1/(2b)}$ follows from the fact that $1-2/x^4 > (2-2^{1/x})\cdot 2^{1/x}$ for $x=2b>2$.
	In other words, we have that $X  \le (1+\tau)\cdot \expectation{\calP_n}{X}$ implies $B_\delta$, and since we assume that $X  \le (1+\tau)\cdot \expectation{\calP_n}{X}$ holds then $B_\delta$ is also true.
	For bounding the probability of $B_\gamma | B_\delta$, we have that:
	\begin{align}\Prr{A}{\neg B_\gamma\middle\vert\, B_\delta} & \leq \exp\parentheses{-m (1-\delta)\gamma^2/4} \leq\exp\parentheses{-m (1-\delta)\gamma^2/4} \leq \exp\parentheses{-\Theta(\gamma^2m)}\;,
	\end{align}
	where $1-\delta = 2-2^{1/2b}\geq 2-\sqrt{2}$ since $b>1$.
	We then get that
	\begin{align*}
		2\gamma &= (1-2^{-1/(2b)})\cdot 2^{1/b} > 1/(2b)^2\;,
	\end{align*}
	since $(1-2^{-1/x})\cdot 2^{2/x}>1/x^2$ for all $x=2b>0$.
	Finally, we have that:
	\begin{align}\Prr{A}{B_\gamma\middle\vert\, B_\delta} & \geq  1 - \exp\parentheses{-\frac{m}{\Theta(\log^4(1/\eps))} }\;.
	\end{align}

\end{document}